\documentclass{article}

\usepackage[final]{neurips_2021}

\usepackage[utf8]{inputenc} 
\usepackage[T1]{fontenc}    
\usepackage{hyperref}       
\usepackage{url}            
\usepackage{booktabs}       
\usepackage{amsfonts}       
\usepackage{nicefrac}       
\usepackage{microtype}      
\usepackage{xcolor}         

\usepackage{amsthm}
\usepackage{amsfonts}
\usepackage{bbm}
\usepackage{graphics}
\usepackage{graphicx}
\usepackage{makecell}
\usepackage{color}
\usepackage{listings}

\newtheorem{thm}{Theorem}
\newtheorem{lem}[thm]{Lemma}

\usepackage{amsmath} 
\usepackage{amssymb}
\usepackage{bm}
\usepackage{ascmac}

\usepackage{algorithm,algpseudocode}

\usepackage{amsmath}

\newcommand{\Real}{\mathbb{R}}
\newcommand{\Natural}{\mathbb{N}}

\newcommand{\Ex}{\mathbb{E}}
\newcommand{\Prob}{\mathbb{P}}
\newcommand{\eps}{\epsilon}
\newcommand{\beps}{\bm{\epsilon}}

\newcommand{\Normal}{\mathcal{N}}
\newcommand{\Tdist}{\mathcal{T}}

\newcommand{\Ind}{\bm{1}}

\newcommand{\Cov}{\mathrm{Cov}}
\newcommand{\boot}{\mathrm{boot}}

\newcommand{\ED}{\mathcal{D}}

\newcommand{\EH}{\mathcal{H}}

\newcommand{\ES}{\mathcal{S}}
\newcommand{\ET}{\mathcal{T}}
\newcommand{\EU}{\mathcal{U}}

\newcommand{\EX}{\mathcal{X}}

\newcommand{\Hboth}{\mathcal{H}}
\newcommand{\Hzero}{\mathcal{H}_0}
\newcommand{\Hone}{\mathcal{H}_1}

\newcommand{\ba}{\bm{a}}
\newcommand{\balpha}{\bm{\alpha}}

\newcommand{\bp}{\bm{p}}
\newcommand{\bx}{\bm{x}}

\newcommand{\bv}{\bm{v}}
\newcommand{\bmu}{\bm{\mu}}
\newcommand{\bbeta}{\bm{\beta}}

\newcommand{\halpha}{\hat{\alpha}}
\newcommand{\hbalpha}{\hat{\bm{\alpha}}}

\newcommand{\bX}{\bm{X}}

\newcommand{\bu}{\bm{u}}

\newcommand{\br}{\bm{r}}

\newcommand{\FDR}{\mathrm{FDR}}
\newcommand{\hatFDR}{\hat{\mathrm{FDR}}}

\newcommand{\hatpi}{\hat{\pi}}

\newcommand{\mtwo}{m_{\mathrm{two}}}
\newcommand{\mone}{m_{\mathrm{one}}}

\newcommand{\Unif}{\mathrm{Unif}}

\allowdisplaybreaks

\title{Controlling False Discovery Rates under Cross-Sectional Correlations}

\author{%
  Junpei Komiyama \\
  Leonard N. Stern School of Business\\
  New York University\\
  New York, USA \\
  \texttt{junpei@komiyama.info} \\
   \And
   Masaya Abe \\
   Nomura Asset Management Co. \\
   Tokyo, Japan \\
   \texttt{masaya.abe.428@gmail.com} \\
   \And
   Kei Nakagawa \\
   Nomura Asset Management Co. \\
   Tokyo, Japan \\
   \texttt{kei.nak.0315@gmail.com} \\
   \And
   Kenichiro McAlinn \\
   Temple University \\
   Fox School of Business \\
   Philadelphia, USA \\
   \texttt{kenichiro.mcalinn@temple.edu} \\
}

\begin{document}

\maketitle

\begin{abstract}
We consider controlling the false discovery rate for testing many time series with an unknown cross-sectional correlation structure. 
Given a large number of hypotheses, false and missing discoveries can plague an analysis. 
While many procedures have been proposed to control false discovery, most of them either assume independent hypotheses or lack statistical power.
A problem of particular interest is in financial asset pricing, where the goal is to determine which ``factors" lead to excess returns out of a large number of potential factors.
Our contribution is two-fold. First, we show the consistency of Fama and French's prominent method under multiple testing. 
Second, we propose a novel method for false discovery control using double bootstrapping.
We achieve superior statistical power to existing methods and prove that the false discovery rate is controlled.
Simulations and a real data application illustrate the efficacy of our method over existing methods.
\end{abstract}

\section{Introduction}
The problem of multiple testing is prevalent in many domains; from statistics to machine learning, and genomics to finance.
One such problem is with false discoveries, where many non-significant hypotheses are reported as significant.
As the number of potential hypotheses can be enormous in many applications, a large number of false hypotheses are expected to be rejected, thus making controlling for the false discovery rate (FDR) necessary.
In the statistics literature, many methods to control the FDR have been proposed.
The most widely used is the Benjamini-Hochberg (BH) method \citep{benjamini1995controlling}, though the best performing procedure is the asymptotic method proposed by  \cite{storey2002} (Storey's method).
However, one key assumption Storey's method makes is independence amongst hypotheses; an assumption that is not realistic in many applications.
Although methods to control the FDR under correlated hypotheses have been proposed, such as \citep{benjamini2001,sarkar2008}, they lack statistical power and are inadequate in many circumstances. 

Our research is motivated by a longstanding problem in the financial asset pricing literature: finding ``factors'' that successfully predict the cross-sectional returns of stocks.
As finding effective factors is important for investment decisions, as well as understanding economic and social organization, much research has been devoted to finding these factors.
Over the decades, the search for new factors has produced hundreds of potential candidates \citep{harvey2016and,mclean2016does,hou2020replicating}, and is often described as a ``factor zoo.''
Although this is a typical multiple testing problem, \cite{harvey2016and} showed that almost all of the past research on factors fail to control for false discoveries.
As a result, there has been a recent interest in FDR control in finance, including \cite{harvey2016and,harvey2020}, which extend the resampling idea of \cite{famafrench2010}.
However, as financial data are highly correlated, its corresponding hypotheses are also correlated, connecting back to the limitation of existing methods.

To deal with problems with the existing methodology in FDR control, and to effectively detect signals in the factor zoo, we propose a novel methodology using null bootstrapping. The contributions of this paper are as follows:
{
\setlength{\leftmargini}{15pt} 
\begin{itemize}
\itemsep0em 
    \item We explain the resampling method of \cite{famafrench2010}, which is used to make valid inference on the percentiles of $p$-values, and derive its consistency (Section \ref{sec:ff}). 
    \item We formalize and generalize this method by proposing a framework for multiple testing with null bootstrapping (Section \ref{sec:setup}).
    \item Using the null bootstrap, we propose a testing method that controls the FDR (Section \ref{sec:method}). Unlike existing finance literature \cite{harvey2020}, we do not require the specification of the proportion of true signals and is fully algorithmic. Our method is valid even under correlated hypotheses, unlike Storey's method \cite{storey2002}, and has stronger power than existing methods (e.g., the YB method \cite{yekutieli1999} and the BY method \cite{benjamini2001}).
    \item We provide extensive simulations and a real-world application to financial factor selection to demonstrate the statistical power of our methodology (Section \ref{sec:sim}).
\end{itemize}
}

\subsection{Related work on multiple testing}

There are two widely studied objectives in multiple testing:
the family-wise error rate (FWER), which is the probability of including a false discovery amongst the found hypotheses, and the false discovery rate (FDR), which is the ratio of false discoveries amongst discoveries. On the one hand, controlling the FWER is crucial in some domains, e.g., when a false discovery can cause severe harm. On the other hand, the FDR provides a stronger statistical power than the FWER. This paper considers the FDR. Most of the existing methods for FDR control, such as Storey's method \citep{storey2002,storey2004} and $\alpha$-investing method \citep{foster2008}, do not work when the hypotheses are correlated\footnote{\cite{ramdas2019} identifies four categories of improvements in multiple testing: 1) Beliefs on the null distribution, 2) importance weighting, 3) grouping hypotheses, and 4) dependency amongst hypotheses. This paper considers and improves upon 1) and 4).}. 
\cite{yekutieli1999} proposed a resampling-based method to control the FDR. Their model (Section 3 therein) is similar to our null bootstrapping (Section \ref{sec:setup}), with regard to the assumption that samples from the null distribution are available. However, the power of their method (Section 4 therein) is limited because they cannot exploit knowledge on the ratio of null hypotheses. 
On the other hand, \cite{romano2008test} proposed a step-up testing method that utilizes bootstrap samples from the (estimated) true distribution and showed that it asymptotically controls the FDR, given good estimates of the order statistics of false discoveries.
In related work, \cite{blanchard2009} proposed a two-stage procedure that controls the FDR under dependence, given knowledge on the dependency (the shape function (3) therein). See also Section 1.2 in \cite{lei2018} and Section 1.2--1.4 in \cite{cai2021} regarding the literature on  testing methods that exploit contextual or structural information. 

\section{Cross-sectional resampling}\label{sec:ff}

Let there be $N$ portfolios (portfolios, here, are constructed using financial factors) and $T$ time steps. We use $i \in [N] = \{1,2,\dots,N\}$ to index a portfolio, and $t \in [T]$ to index a time step. 
We consider Fama and French's three-factor (FF3) model \citep{fama1993common}:
\begin{equation}\label{FF3}
    r_{i,t} = a_i + MKT_t\,b_i  + SMB_t\,s_i + HML_t\,h_i + \varepsilon_{i,t},
\end{equation}
where $r_{i,t}$ is the return of portfolio $i$; $MKT_t$, $SMB_t$, and $HML_t$ are market, size, value factor returns at time $t$; $a_i$ is the ``alpha'' of portfolio $i$, and $\varepsilon_{i,t}$ is the zero-mean noise. 
This model serves as a benchmark against each portfolio, with $a_i$ being the parameter of interest (what is tested). Here, $a_i > 0$ implies that the portfolio outperforms the baseline FF3.

For the ease of exposition, let $\bx_t = (1, MKT_t, SMB_t, HML_t)^\top$ and $\bbeta_i = (a_i, b_i, s_i, h_i)^\top$. Let $\br_i = (r_{i,1},r_{i,2},\dots,r_{i,T})^\top$ and $\beps_i = (\eps_{i,1},\eps_{i,2},\dots,\eps_{i,T})^\top$.
The matrix form of Eq.\,\eqref{FF3} is denoted as 
\begin{equation}\label{FF3_matrix}
\br_i = \bX \bbeta_i + \beps_i,
\end{equation}
where the $t$-th row of Eq.~\eqref{FF3_matrix} corresponds to Eq.~\eqref{FF3}.

The cross-sectional bootstrap proposed by \cite{famafrench2010} consists of the following steps. First, we obtain the ordinary least squares (OLS) estimate, $\hat{\bbeta}_i = (\bX^\top \bX)^{-1} \bX^\top \br_i \in \Real^{1+3}$, for each portfolio $i$. The parameter of the interest is the first component $\hat{a}_i = (\hat{\bbeta}_i)_1$ of $a_i$. 
The second step is to obtain the cross-sectional dependence of the portfolios by bootstrap simulations. 
Let the residual be $\hat{\eps}_{i,t} = (r_{i,t} - \bX_t^\top \hat{\bbeta}_i)$. 
For each simulation $b=1,2,\dots,B$, we sample (with replacement) $T$ time steps $\ET_b = (t_1^{(b)},t_2^{(b)},\dots,t_T^{(b)}) \in [T]\times[T]\times...\times[T]$ to obtain a variation of residual
\begin{equation}
\hat{e}_i^{(b)} = \frac{1}{T} \sum_{t \in \ET_b} \hat{\eps}_{i, t}.
\end{equation}
The correlation of $({\hat{e}}_i^{(b)})$, among portfolios $[N]$, measures the cross-sectional dependence of the portfolio returns; for two portfolios $i,j \in [N]$ that adopt a similar investment strategy, the variations ${\hat{e}}_i^{(b)}$ and ${\hat{e}}_j^{(b)}$ are strongly correlated.

Although the original motivation in \cite{famafrench2010} is mainly on percentile inference, we use this bootstrap method differently. We consider the multiple testing problem of finding significantly good (i.e., $a > 0$) portfolios.

\subsection{Consistency of \cite{famafrench2010}}\label{subsec_consistency}

Despite the popularity of this method, its theoretical property has not been fully addressed in the literature of finance. In this subsection, we characterize the property of this estimator. 

\begin{itemize}
\item \textbf{Assumption 2.1 (no perfect collinearity):} Matrix $\bX$ has full row rank $(= 4)$.
\item \textbf{Assumption 2.2 (no serial correlation):} Conditional on $\bX$, the error $(\eps_{i,t})_{t \in [T]}$ is independent and identically distributed as $\Normal(0, \sigma_i^2)$. 
\item \textbf{Assumption 2.3 (zero empirical bias):} $(1/T) \sum_t \bx_t = (1,0,0,0)$.
\end{itemize}
Assumption 2.1 and 2.2 are standard assumptions in square regression. While Assumption 2.2 states that the noise term is i.i.d., with respect to the temporal direction, we allow $\eps_{i,t}$ and $\eps_{j,t}$ to be correlated at each time step $t$.
Assumption 2.3 is easily satisfied by normalizing non-intercept features and is consistent with the arbitrage pricing theory (c.f. Eq.~(2) in \cite{ross1976arbitrage}), which is the theoretical background of the multi-factor model in finance.

Under the assumptions above, we derive the normality of the OLS statistics.
\begin{thm}{\rm (Marginal distribution)}\label{thm_marginaldist}
Let Assumptions 2.1--2.3 hold. Then, conditional on $\bX$, the OLS estimator $\hat{\bbeta}_i$ of $\bbeta_i$ is distributed as $\Normal(\bbeta_i, \sigma_i^2 (\bX^\top \bX)^{-1})$. Moreover, for $i,j \in [N]$, 
\begin{align}
\Ex[\hat{a}_i] &= a_i\\
T \Cov(\hat{a}_i, \hat{a}_j) &= \Cov(\eps_{i,1}, \eps_{j,1}).
\end{align}
\end{thm}
Due to page limitations, all the proofs are in the supplementary material.

The following theorem states that the bootstrap samples captures the cross-sectional correlation of the estimator $(\hat{a}_i)_{i \in [N]}$.
\begin{thm}{\rm (Consistency of residual bootstrap)}\label{thm_resboot}
Let Assumptions 2.1--2.3 hold. Conditioned on $\EX$, we have
\begin{equation}\label{ineq_bootstrap_mean}
 \Ex_{\boot}[ \hat{e}_i^{(b)}] = 0,
\end{equation}
where $\Ex_{\boot}$ is the expectation over the randomness of bootstrap resampling\footnote{The formal definition of $\Ex_{\boot}$ is given in Eq.~\eqref{ineq_defboot} in the supplementary material.}.
Moreover, with probability at least $1/\delta$, we have
\begin{equation}\label{ineq_bootstrap_cov}
 T\, \Cov_{\boot}(\hat{e}_i^{(b)}, \hat{e}_j^{(b)}) = \Cov(\eps_{i,1}, \eps_{j,1}) + O\left(\frac{\log(1/\delta)}{T}\right)
\end{equation}
where 
$
\Cov_{\boot}(\hat{e}_i^{(b)}, \hat{e}_j^{(b)}) 
= \Ex_{\boot}\left[ \hat{e}_i^{(b)} \hat{e}_j^{(b)} \right]
$
is the covariance over the resampling.
\end{thm}
Theorem \ref{thm_resboot} states that bootstrap samples simulate the first and second moments, given sufficiently large $T$. 

\section{Multiple testing with cross-sectional bootstrap}\label{sec:setup}

This section introduces our setup for multiple testing over portfolio returns. Each of the portfolios, representing some financial factor, corresponds to a statistical hypothesis; on whether $a_i = 0$ (no significant alpha) or not. 

\subsection{Data generating process}

The estimated alpha $\hat{a}_i$ for each portfolio $i \in [N]$ is normally distributed, which yields the studentized test statistic $\halpha_i = \sqrt{T}\hat{a}_i/\hat{\sigma}_i$, where $\hat{\sigma}_i^2$ is the sample error variance. Let $\hbalpha = (\halpha_1,\halpha_2,\dots,\halpha_N)$.
Let $\alpha_i = \Ex[\sqrt{T}a_i/\sigma_i]$ be the rescaled true alpha of portfolio $i$ and its vector be $\balpha = (\alpha_1,\alpha_2,\dots,\alpha_N)$. 
Let the residual be $u_i = \alpha_i - \halpha_i$ and $\bu = (\alpha_1-\halpha_1, \alpha_2-\halpha_2,\dots,\alpha_N-\halpha_N)$. Theorem \ref{thm_marginaldist} implies that each $u_i$ follows a student $t$ distribution with $T-3-1$ degrees of freedom. 
Let $\EU$ be the ``correlated null'' distribution of $\bu$.
There is cross-sectional correlation on the test statistics (i.e., $u_i$ and $u_j$ for $i \ne j$ are correlated). 
Let $u^{(b)}_i = \sqrt{T}\hat{e}_i^{(b)}/\hat{\sigma}_i^{(b)}$ where $\hat{\sigma}_i^{(b)}$ is a bootstrap counterpart of $\hat{\sigma}_i$. Let $\bu^{(b)} = (u^{(b)}_1,u^{(b)}_2,\dots,u^{(b)}_N)$.
Theorems \ref{thm_marginaldist} and \ref{thm_resboot} imply that each bootstrapped sample $\bu^{(b)}$ simulates $\EU$, for sufficiently large $T$. For the ease of discussion, we assume the availability of an infinite number of samples from $\EU$. Namely, 
$
\bu^{(b)} \sim \EU
$
for each $b=1,2,3,\dots,B$.

We derive the asymptotic property of the cross-sectional bootstrap with the following assumptions
\begin{itemize}
\item \textbf{Assumption 3.1:} The observed statistic is $\hbalpha = \balpha + \bu \in \Real^N$, where $\bu \sim \EU$. The marginal distribution of $u_i$ follows a student-$t$ distribution with degrees of freedom $T-4$.
\item \textbf{Assumption 3.2:} Bootstrap samples $\bu^{(1)},\bu^{(2)},\bu^{(3)},\dots$ are  i.i.d. samples from $\EU$.
\end{itemize}
Theoretical results in the subsequent sections utilize these assumptions.

\subsection{Marginal distribution, $p$-values, null, and alternative hypotheses}

A $p$-value function $\bp = \bp(\hbalpha) = (p_1(\alpha_1),p_2(\alpha_2),\dots,p_N(\alpha_N))$ is the tail area of the $t$-distribution. In the case of one-sided testing, we have
\begin{equation}
p_i(u) = \int_{u}^\infty f(u') du',
\end{equation}
and in the case of two-sided testing, we have
\begin{equation}
p_i(u) = 2 \int_{|u|}^\infty f(u') du',
\end{equation}
where $f(u')$ is the probability density function (PDF) of the $t$-distribution. By definition, 
\begin{equation}
\forall p\in(0,1), \Prob[p_i(u_i) \le p ] = p
\end{equation}
for $u_i \sim \Tdist_{k}$. 
Let $\Hboth = [N]$ denote all the $N$ hypotheses.
The null hypotheses are
$\Hzero = \Hzero(\balpha) = \{i \in [N]: \alpha_i \le 0\}$
for the one-sided testing or
$\Hzero = \Hzero(\balpha) = \{i \in [N]: \alpha_i = 0\}$ for the two-sided testing, respectively. The alternative hypotheses are $\Hone(\balpha) = \Hboth \setminus \Hzero$. Let $N_0,N_1 = |\Hzero|,|\Hone|$, and $\pi_0, \pi_1 = N_0/N, N_1/N$. 

\subsection{Multiple testing}

A multiple testing method is a function:
$R(\hbalpha) : \Real^N \rightarrow 2^{[N]}.$
Namely, given the observation vector $\hbalpha$, it returns a subset of hypotheses that likely belong to $\Hone$. 
We say that the algorithm rejects hypothesis $i$ if $i \in R(\hbalpha)$. Let $\EH_{1|0}, \EH_{1|1}, \EH_{0|1}, \EH_{0|0}$ denote false positive, true positive, false negative, and true negative, respectively. Let $R = |\Hone| = N_{1|0} + N_{1|1}$  and let $N_{1|0}, N_{1|1}, N_{0|1}, N_{0|0}$ denote corresponding cardinalities. 
The false discovery rate, $\FDR = \FDR(\ba)$, which is the ratio of false positives over all the positives:
\begin{equation}
\FDR(\balpha) = 
\Ex_{\bu \sim \EU}\left[
\frac{N_{1|0}}{\max(R,1)}
\right],
\end{equation}
where the expectation is taken with respect to the randomness of $\bu \sim \EU$, such that $\hbalpha = \balpha + \bu$. 
The goal here is to develop a testing method of maximum statistical power such that $\FDR \le q$ for a given $q \in (0,1)$ and for any $\balpha$. 

\begin{figure*}[tb]
\vspace{0.5em}
\centering
\includegraphics[bb = 0 0 407 81,scale=0.9]{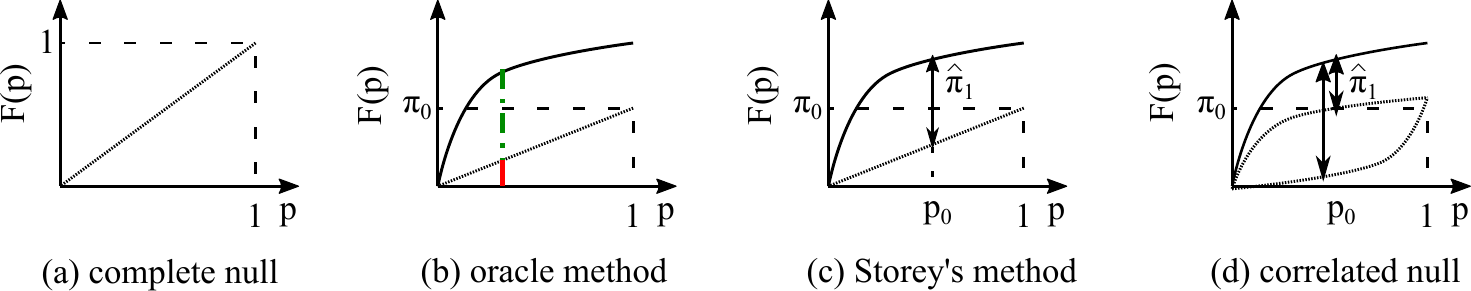}
\caption{Illustration of the oracle method \cite{genovese2004} and the plug-in estimator by Storey. (a) Independent null hypotheses are uniformly distributed, and thus its CDF is $F(p) = p$ (when $N \rightarrow \infty$). (b) Oracle method: given knowledge of $\pi_0 = N_0/N$, the red line corresponds to the density of nulls, whereas the green dash-dotted line corresponds to the density of the alternative hypotheses. Oracle threshold is $p_{\mathrm{oracle}}$, such that the ratio of red and green is $\delta$:$1-\delta$. (c) Storey's method adopts a plug-in estimator $\hatpi_1$ to the oracle formula. (d) When the hypotheses are correlated, the variance of $F(p)$ is unignorable, and Storey's method no longer estimates $\hatpi_1$ correctly.
}
\label{fig:methods}
\vspace{-1em}
\end{figure*}

\subsection{Step-up method and Storey's method}\label{subsec:bhstorey}

One of the most widely accepted method for controlling the FDR was proposed by Benjamini and Hochberg \citep[BH:][]{benjamini1995controlling}, which adopts a step-up procedure. The BH method sorts $p$-values, $p_{(1)} \le p_{(2)} \le p_{(3)} \le \dots \le p_{(N)}$, and rejects hypotheses in an increasing order up to 
\begin{equation}\label{ineq:bh}
p = \max_i \{ q i/N:  p_{(i)} \le q i / N\}.
\end{equation}
This method always controls the FDR at level $\pi_0 q \le q$ if the $p$-values are independent or positively correlated in specific ways\footnote{Correlation needs to satisfy the \textit{positive regression dependence on subset} (PRDS) condition \cite{benjamini2001}.}. When the dependence structure of $p$-values are unknown, a logarithmic correction is required \citep{bky2006}. Replacing $q$ with $q /(\log N + (1/2))$ in Eq.~\eqref{ineq:bh} guarantees the same FDR under arbitrary correlation\footnote{This logarithmic correction is tight: See Lecture 9 in \cite{candes2018tutorial} for construction of a worst-case example.}.

Although the BH method with the logarithmic correction strictly controls the FDR, for most cases, the threshold of Eq.~\eqref{ineq:bh} is excessively conservative, which results in weak statistical power compared to the ideal threshold. Under independent $p$-values, an optimal oracle method was proposed in \cite{storey2002,genovese2004}, which gives an accurate estimate of the FDR. Figure \ref{fig:methods} illustrates the oracle method. The oracle method requires the ratio of true hypotheses $\pi_1$, which is hard to obtain {\it a priori}. 
Storey's  method \citep{storey2002} adopts a plug-in estimator 
\begin{equation}
\hatpi_1 = \max\left( 
0, \frac{\hat{F}(p_0) - p_0}{1-p_0}
\right),
\end{equation}
for some reference point $p_0 \in (0,1)$. 
Given $\hatpi_1$, a threshold 
\begin{equation}\label{ineq:storeythreshold}
p = \sup\{p': \hat{F}_{\hbalpha}(p') \ge (1-\hatpi_1)p'/q\},
\end{equation}
is implemented, where $\hat{F}_{\hbalpha}(p) = (1/N) \sum_{i \le N} \Ind[p_i(\halpha_i) \le p]$ is the empirical CDF of $(p_i(\halpha_i))_{i\in[N]}$.
Storey's method controls the FDR under independent hypotheses \citep{storey2004,genovese2004}.
However, Storey's method no longer controls the FDR, both in theory and practice \citep{bky2006,romano2008test} under correlation.

\section{Our proposed method: DDBoot}\label{sec:method}

  \begin{algorithm}[!t]
    \caption{Dueling Double Bootstrap (DDBoot)}
    \label{alg_proposed}
    \begin{algorithmic}
        \Require $\hbalpha$, $q \in (0,1)$, $V$, $W \in \Natural^+$.
        \Ensure Threshold $p \in (0,1)$.
        \For{$v=1,\cdots,V$}
            \State Sample $\hbalpha^{(v)}$ by using Algorithm \ref{alg_proposed_sub_bmu}.
            \State Find $c_q^{(v)} \leftarrow \sup\{c_q\in(0,1): \hatFDR(\hbalpha^{(v)}, c_q) \le q/2\}$ by using Algorithm \ref{alg_proposed_sub} and binary search.
        \EndFor
        \State Determine final threshold $p$ by Eq.~\eqref{ineq:propsup} with $c_q \leftarrow \min_{v=1}^{V} c_q^{(v)}$.
        \end{algorithmic}
  \end{algorithm}
  \begin{algorithm}[!t]
    \caption{Subroutine for generating $\hbalpha^{(v)}$}
    \label{alg_proposed_sub_bmu}
    \begin{algorithmic}
        \Require $\hbalpha$.
        \Ensure $\bmu$.
        \State $\bu^{(v)} \sim \EU$.
        \State $\hbalpha^{(v)} \leftarrow \hbalpha - \bu^{(v)}$. 
        \If{Two-sided testing}
            \State $\halpha^{(v)}_i \leftarrow 0$ for all $\{i \in [N]: |\halpha_i| \le |u_i^{(v)}|\}$.
        \EndIf
        \end{algorithmic}
  \end{algorithm}

  \begin{algorithm}[!t]
    \caption{Subroutine for calculating $\hatFDR(\hbalpha^{(v)}, c_q)$}
    \label{alg_proposed_sub}
    \begin{algorithmic}
        \Require $c_q \in (0,1)$, $\bmu$, $W$.
        \For{$w=1,\cdots,W$}
            \State $\bu^{(w)} \sim \EU$.
            \State $\hbalpha^{(w)} = \hbalpha^{(v)} + \bu^{(w)}$.             \State Calculate threshold $p \leftarrow \sup\{p': \hat{F}_{\hbalpha^{(w)}}(p') \ge p'/c_q\}$. 
            \State Calculate number of discoveries: $R \leftarrow |\{i\in[N]: p_i(\halpha^{(w)}_i) \le p\}|$, $N_{1|0} \leftarrow |\{i\in[N]: p_i(\halpha^{(w)}_i) \le p, i \in \Hzero(\hbalpha^{(v)})\}|$ \Comment{FDR when $\Hzero = \Hzero(\hbalpha^{(v)})$ are nulls.}
            \State Compute $\hatFDR^{(w)} \leftarrow \frac{N_{1|0}}{\max(1,R)}$. 
        \EndFor
        \State $\hatFDR(\hbalpha^{(v)}, c_q) \leftarrow \frac{1}{W} \sum_{w=1}^{W} \hatFDR^{(w)}$.
    \end{algorithmic}
  \end{algorithm}

We consider a threshold that belongs to the same class as Storey's method. Namely, we reject all the hypotheses $p_i(\hbalpha) \le p$, such that 
\begin{equation}\label{ineq:propsup}
p = \sup\{p': \hat{F}_{\hbalpha}(p') \ge p'/c_q\},
\end{equation}
where the correction factor $c_q = c_q(q)$ is the factor that we aim to obtain. 
Let $\FDR(\balpha, c_q)$ denote the true FDR of the threshold of Eq.~\eqref{ineq:propsup}. 
Roughly speaking, when all the null hypotheses are independently distributed and under the complete null (i.e., $\hatpi_1 = 0$), we expect $c_q \approx q$ to control $\FDR(\balpha, c_q) \le q$. 
A larger value of $\hatpi_1$ allows a larger $c_q$. In an extreme case, if $\pi_1 = 1$ (i.e., complete alternative hypotheses), $c_q$ can be arbitrarily large because there is no false positive. 
While the true parameter $\balpha$ is not observable, we have access to $\hbalpha = \balpha + \bu$, where $\bu \sim \EU$, as well as bootstrap samples, $\bu^{(1)},\bu^{(2)},\bu^{(3)},...$, which are drawn from the same null distribution $\EU$ (Assumptions 3.1 and 3.2). 

The proposed Algorithm \ref{alg_proposed} involves two subroutines. Namely, it 1) samples $\hbalpha^{(v)}$ and 2) estimates $\FDR$ given $\hbalpha^{(v)}$. By using these subroutines, Algorithm \ref{alg_proposed} defines the threshold $c_q$ such that its estimated FDR is controlled for all sampled parameters $(\hbalpha^{(v)})_{v \in [V]}$.

{\bf 1) Sampling of $\hbalpha^{(v)}$ (Algorithm \ref{alg_proposed_sub_bmu}):}
By definition, $\balpha = \hbalpha - \bu$, and thus a reasonable estimate, $\hbalpha^{(v)}$ of $\balpha$, is $\hbalpha^{(v)} = \hbalpha - \bu^{(v)}$, with a bootstrap sample $\bu^{(v)} \sim \EU$. 

We design these samples $\hbalpha^{(v)}$ so that the number of null hypotheses, $|\Hzero(\hbalpha^{(v)})|$, is lower-bounded in terms of the number of true null hypotheses, $N_0 = |\Hzero(\balpha)|$, because 
the more null hypotheses there are, the larger the FDR is likely to be \footnote{\cite{sarkar2006} formalizes the conditions where the number of the null characterizes the hardness: Eq.~(3.2) therein states that the hardest case is the most marginal case where all the null and alternative hypotheses are close to the margin (i.e., $\hbalpha \approx 0$ in our notation), and on the margin, the hardness of the instance is solely defined by the number of null hypotheses. Although this argument in \cite{sarkar2006} essentially requires the independence of the hypotheses, the number of null hypotheses reasonably characterizes the hardness of the instance even under correlated hypotheses.}. 

In the case of one-sided testing, the naive value of $\hbalpha^{(v)}$ contains at least $N_0/2$ null hypotheses. 
In the case of two-sided testing, due to the continuity of the distribution, the points including the null hypothesis, $|\{\hbalpha^{(v)}: \halpha^{(v)}_i = 0\}| > 0$, has measure zero, which results in an underestimation of $N_0$. 
By setting $\halpha^{(v)}_i = 0$ for all $i$, such that $|\halpha_i| < |u^{(v)}_i|$, we guarantee that at least $\hbalpha^{(v)}$ contains at least $N_0/2$ null hypotheses.
The following theorem states that, with probability $1-O(1/V)$, we have at least one $\hbalpha^{(v)}$ that contains at least $N_0/2$ null hypotheses.
\begin{lem}{\rm (Number of null hypotheses)}\label{lem:numhypo}
With probability at least $1 - \frac{2}{V + 1}$, we have
\begin{equation}\label{ineq:numhypo}
\exists{v \in [V]}\ |\Hzero(\hbalpha^{(v)})| \ge N_0/2
\end{equation}
for both one-sided and two-sided testing.
\end{lem}%
The proof, which is in the supplementary material, introduces the ``dueling estimator'' $m(\bu^{(v)}, \hbalpha)$, via which we lower-bound the number of null hypotheses.  

{\bf 2) Estimating $\hatFDR(\hbalpha^{(v)}, c_q)$ (Algorithm \ref{alg_proposed_sub}):}
Given $\hbalpha^{(v)}$ that has at least $N_0/2$ null hypotheses, we optimize $c_q$, such that $\FDR(\hbalpha^{(v)}, c_q) \approx q/2$.
The estimated FDR is obtained by replacing the expectation by its empirical counterpart with $W$ bootstrapped samples. The following theorem provides the error of the empirical FDR.
\begin{lem}{\rm (Accuracy of the empirical FDR)}\label{lem_empfdr}
For each $\hbalpha^{(v)}$ and $\delta \in (0,1)$, with probability at least $1-2\delta$, we have
\begin{equation}\label{ineq:empfdr}
|\FDR(\hbalpha^{(v)}, c_q) - \hatFDR(\hbalpha^{(v)}, c_q)| \le \sqrt{\frac{\log(1/\delta)}{2 W}}.
\end{equation}
\end{lem}
The proof of Lemma \ref{lem_empfdr} directly follows from the Hoeffding inequality, since $(\hatFDR^{(w)})_{w=1,2,\dots,W}$ are unbiased estimators of $\FDR$ bounded in $[0,1]$ and independent to each other. 
Finally, one can use a binary search of the optimal threshold of $c_q$, such that
$ 
\inf_{c_q\in(0,1)} \hatFDR(\hbalpha^{(v)}, c_q) \le q/2,
$ 
assuming the monotonicity of the FDR.

The following theorem summarizes the accuracy of the proposed algorithm.
\begin{thm}{\rm (Main theorem)}\label{thm:main_prac}
For any $\balpha'$, assume that $\FDR(\balpha', c)$ is non-decreasing in $c$.
Let $c_q^* = \sup_{c_q} \{c_q: \FDR(\balpha, c_q) \le q\}$ be the desired threshold for the true FDR of level $q$.
If there exists $\hbalpha^{(v)}: v \in [V]$, such that 
\begin{equation}\label{ineq:assm_existhard}
\FDR(\hbalpha^{(v)}, c_{q'}^*) \ge q'/2
\end{equation}
for all $q' \in (0,1)$,
then, with probability at least $1-2\delta$,
the proposed method controls $\FDR$ at most
\begin{equation}
q + 2 \sqrt{\frac{\log((V S)/\delta)}{2 W}},
\end{equation}
where $V$ is the number of samples of $\hbalpha^{(v)}$, $S$ is the number of repetitions in the binary search, and $W$ is the number of bootstrap samples to estimate $\hatFDR$, respectively.
\end{thm}%
Theorem \ref{thm:main_prac} states that, if there is at least one instance that is harder than the true instance (Eq.~\eqref{ineq:assm_existhard}), and $W$ is sufficiently large, the proposed method controls the FDR. 

\begin{table*}[]
\vspace{-0.5em}
\caption{Results of synthetic simulations. $\pm$ indicates two-sigma confidence intervals of FDRs. Due to space limitation. we omit the results of BY, which is outperformed by the other algorithms.
\textcolor{red}{Red} characters are used for FDRs larger than $q = 0.05$, which indicates failure of controlling the FDR.
}
\label{tbl_results_synth}
\begin{footnotesize}
\begin{center}
\begin{tabular}{llllllllllll}
                & \multicolumn{2}{l}{LSU}                                                 &  & \multicolumn{2}{l}{Storey} &  & \multicolumn{3}{l}{Bootstrapping} &  & Fixed  \\ \cline{2-3} \cline{5-6} \cline{8-10} \cline{12-12} 
                & BH                   & BKY &  & Storey         & Storey-A         &  & YB           & \textbf{DDB} & \textbf{DDBA}          &  & Single \\ \hline
\\
\multicolumn{5}{l}{Scenario 1: $\sigma_{i,j}=0.0, \pi_0 = 0.5$} & & & & & & & \\ \hline
Thr-p&0.0005&0.0005&&0.0008&0.1093&&0.0001&0.0004&0.0018&&0.0385\\
\# of Rej&0.79&0.78&&1.01&7.04&&0.43&0.71&1.45&&6.29\\
FDR&0.0260&0.0256&&0.0344&\textcolor{red}{0.1006}&&0.0164&0.0232&0.0499&&\textcolor{red}{0.1922}\\
&$\pm$0.0060&$\pm$0.0060&&$\pm$0.0065&$\pm$0.0094&&$\pm$0.0052&$\pm$0.0057&$\pm$0.0072&&$\pm$0.0073\\ \hline

\\
\multicolumn{5}{l}{Scenario 2: $\sigma_{i,j}=0.5, \pi_0 = 0.5$} & & & & & & & \\ \hline
Thr-p&0.0012&0.0021&&0.0190&0.1591&&0.0004&0.0022&0.0069&&0.0363\\
\# of Rej&1.48&1.67&&3.51&10.41&&0.92&1.64&2.85&&6.40\\
FDR&0.0243&0.0270&&0.0475&\textcolor{red}{0.1150}&&0.0183&0.0246&0.0440&&\textcolor{red}{0.2040}\\
&$\pm$0.0054&$\pm$0.0056&&$\pm$0.0069&$\pm$0.0097&&$\pm$0.0051&$\pm$0.0053&$\pm$0.0068&&$\pm$0.0137\\ \hline

\\

\multicolumn{5}{l}{Scenario 3: $\sigma_{i,j}=0.9, \pi_0 = 0.5$} & & & & & & & \\ \hline

Thr-p&0.0018&0.0041&&0.0645&0.1924&&0.0019&0.0052&0.0123&&0.0273\\
\# of Rej&2.05&2.51&&8.62&15.30&&2.12&2.78&4.41&&6.15\\
FDR&0.0138&0.0179&&\textcolor{red}{0.0866}&\textcolor{red}{0.1544}&&0.0185&0.0220&0.0368&&\textcolor{red}{0.0971}\\
&$\pm$0.0039&$\pm$0.0042&&$\pm$0.0086&$\pm$0.0105&&$\pm$0.0052&$\pm$0.0049&$\pm$0.0061&&$\pm$0.0119\\ \hline
\end{tabular}
\end{center}
\end{footnotesize}
\end{table*}%

\section{Simulation and application}\label{sec:sim}

We conduct a comprehensive set of simulations to evaluate the empirical performance of the proposed method: DDBoot (Algorithm \ref{alg_proposed}).
Section \ref{subsec:simsynth} describes the results with a synthetic data generating process.
Section \ref{subsec:simreal} describes results from the performance of factors in the U.S. equity market. All the tests are two-sided. All the FDR targets are set at $q = 0.05$.
The source code of the implemented methods is publicly available\footnote{\url{https://www.dropbox.com/s/k7rnymz78tnvlh5/arxiv_jun2021_fdr_materials.zip?dl=0}}.

\noindent\textbf{Methods:} We compare and test the following methods:
``Single,'' adopts a fixed threshold $p=0.05$, ignoring multiplicity; linear step-up (LSU) methods, including ``BH'' \citep{benjamini1995controlling}, ``BY'' \citep{benjamini2001}, and ``BKY'' \citep{bky2006}; Storey methods, including ``Storey,'' which adopts a fixed reference point $p_0 = 0.5$, and ``Storey-A,'' which adaptively determines the reference point $p_0$ by bootstrapping $p$-values \citep{storey2004}; bootstrapping methods, including ``YB'' \citep{yekutieli1999}, ``DDB'' (DDBoot: Algorithm \ref{alg_proposed}), and ``DDBA,'' which is an aggressive version of DDBoot, where the threshold is controlled at $q$, instead of $q/2$.
We adopt pointwise $1-q/2$ upper percentile for $r_\beta^*(p)$ in the YB method and set its confidence level $q/2$ so that its FDR is less than $q/2 + q/2 = q$. The number of bootstrapped samples in YB and Storey-A are set to $500$. We set $V, W = 20, 500$ for DDB so that $V = 1/q, W\ge 1/q^2$.

\subsection{Synthetic data}\label{subsec:simsynth}

We generate multivariate $t$-statistics following the standard procedure (see section 7.1 in \cite{romano2008test}). We set the degrees of freedom as $100$ and the number of hypotheses as $N = 50$. 
The null distribution $\EU$ is a multivariate $t$-distribution with its covariance $\Sigma = (\sigma_{i,j}) \in \Real^{N \times N}$. The diagonals are  $\sigma_{i,i} = 1$, for all $i$. We consider several types of correlation for the non-diagonal $\sigma_{i,j}$. The true $\mu_i$ for each alternative hypothesis $i \in \Hone$ is drawn independently from $2 \times \Unif(0,1)$. 

\noindent\textbf{Scenarios:}
We consider two variables: 
The null ratio $\pi_0 = 0.25, 0.5, 1.0$ and the correlation $\sigma_{i,j} = 0.0, 0.5, 0.9$.
This yields $3 \times 3 = 9$ scenarios, which we only report the case for $\pi_0 = 0.5$ (see appendix for the full table). 

\noindent\textbf{Results:}
Table \ref{tbl_results_synth} shows the results of the simulations. 
All results are averaged over 2,000 runs. 
``Thr-p'' (threshold $p$) is the largest rejected $p$-value, ``\# of Rej'' is the number of rejected hypotheses $|R|$ that indicates the statistical power. ``FDR'' is the false discovery rate of each method.
Storey, Storey-A, and Single have the largest statistical power, though they fail at controlling the FDR under correlated hypotheses. 
The other six methods successfully control the FDR. Among them, DDBA has the strongest statistical power by a large margin, followed by DDB, BKY, and BY. 
Note that, when $\pi_0 = 1$ (complete null), there are some cases where DDBA fails to control the FDR (see Table 4 in the appendix), though it is still robust compared to Storey's method.

\subsection{Real data}\label{subsec:simreal}

\begin{figure}
\vspace{-1em}
\centering
\includegraphics[scale=0.6]{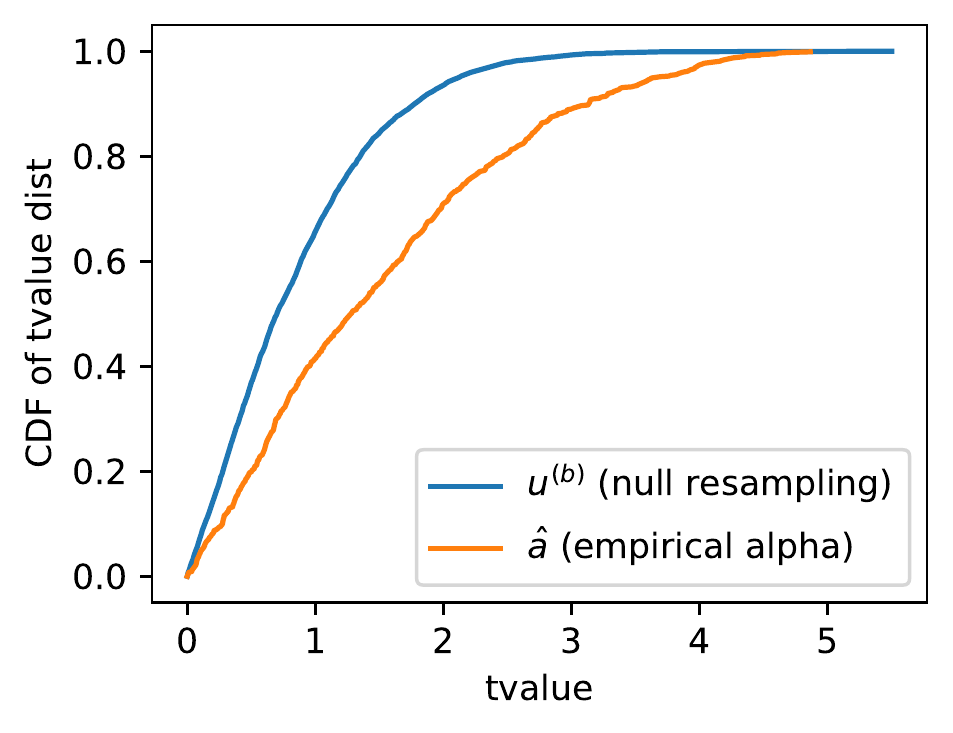}
\vspace{-1em}
\caption{Cumulative distribution of $(\halpha_i)_{i\in[N]}$ (alphas) and $(u_i^{(b)})_{i\in[N],b\in[B]}$ (resampled residuals). Significant portion of alpha surpasses the bootstrapped null samples, and thus the estimator of $\pi_1$ is expected to boost the statistical power. }
\label{fig:tvalue_cdf}
\vspace{-1em}
\end{figure}%
\begin{figure*}[t!]
\centering
\includegraphics[scale=0.6]{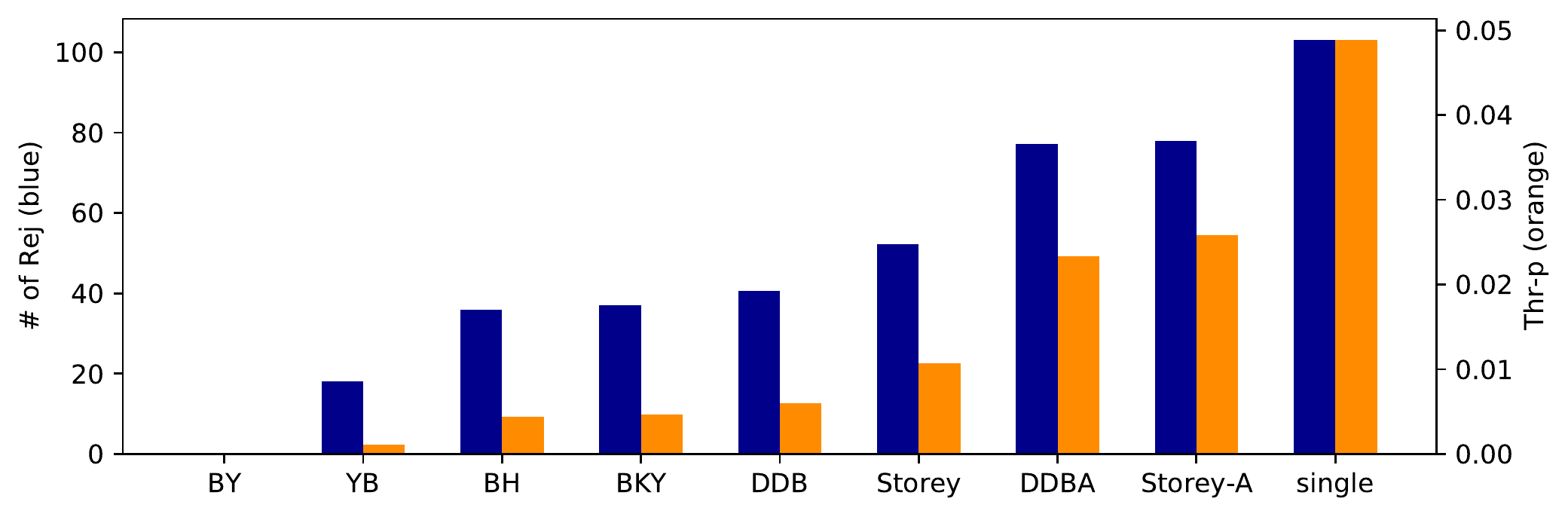}
\vspace{-0.5em}
\caption{Results of simulations on U.S. equity market data. Bars are \# of rejected hypotheses (left, blue) and threshold $p$-value (right, orange). The threshold $p$-value is defined as the largest $p$-value amongst rejected hypotheses. Methods are sorted by their statistical power.}
\label{fig_results_real}
\vspace{-0.5em}
\end{figure*}%

We compare the performance of each method with real-world asset pricing data. 

\textbf{Constructing alphas:} We prepare twenty factors listed in Table 3 in the appendix. Following \cite{yan2017fundamental}, we create a composite signal of these factors by using standard arithmetic operations (e.g., factor $X+Y$ made of factors $X,Y$) to yield $760$ composite factors.  
For each of these $20+760$ factors, we create a quintile long-short portfolio on the stocks listed on the MSCI USA Index from January 2000 to June 2020. Given each portfolio, we estimate the alpha of Eq.~\eqref{FF3} by using the FF3 factors. 
As a result, we obtain empirical alphas ($t$-values), $(\hbalpha_i)_{i\in [N]}$, with $N=780$.
Standard Newey-West corrections are applied to these $t$-values. 
We also obtain $B = 10,000$ samples from the null distribution, $\bu^{(1)},\bu^{(2)},\dots,\bu^{(B)}$, by applying OLS for the zero-alpha return, $r_{i,t}-\hat{a}_i$. 
Figure \ref{fig:tvalue_cdf} shows the distribution of alphas compared to bootstrapped null samples, with a significant portion improving over the null.

\noindent\textbf{Results:}
To speed up computation, we randomly pick half of the 780 factors and conduct multiple testing for the $t$-values. 
Figure \ref{fig_results_real} shows the results. All the results are averaged over 100 runs. 
DDBA performs close to Storey-A. Among the methods that always control the FDR (Tables \ref{tbl_results_synth} and 4), DDB performed the best.

\subsection{Summary of simulation and application}

1) DDB exactly controls at level $q$ and provides stronger statistical power compared to existing methods, such as BY and YB. It also outperforms BH and BKY, \if0\footnote{The improvement of the statistical power by DDB over BH and BKY is a significant improvement given BH and BKY is ``quite hard to break in practice'' (37p in \cite{genovese2004tutorial}).}\fi which do not control the FDR in some corner cases \citep{candes2018tutorial}. 
2) DDBA controls the FDR unless $\pi_0$ is close to $1$. Its performance on real data is comparable to Storey and Storey-A. In contrast, DDBA is robust to correlation. 

In summary, DDBA is the best solution to control for the FDR, unless $\pi_0 \approx 1$. When $\pi_0 \approx 1$ and the hypotheses may be correlated, DDB is a robust solution. When there is no correlation, the well-known Storey's method is the recommended solution.

\section{Conclusion}

Inspired by the residual bootstrapping of \cite{famafrench2010}, we formalize the problem of controlling the FDR among many portfolios. 
We prove the consistency of the original method and propose a novel method for multiple testing, DDBoot, that utilizes samples from the null distribution. 
As the proposed framework exploits the cross-sectional dependence, it can be applied to many other problems where the hypotheses are correlated. Addressing missing data broadens its applications.

\textbf{Limitations:} We have assumed no serial correlations (Assumption 2.2) and zero-mean factors (Assumption 2.3). Moreover, we have used asymptotic property of the test statistics (Assumption 3.1 and 3.2). We remark on these in Section \ref{sec_liminations} in the appendix.

\clearpage
\bibliography{sample}
\bibliographystyle{apalike}

\clearpage
\appendix

\begin{table*}
\centering
\caption{List of the factors. ``Category'' column indicates the corresponding category to Deutsche Bank Quantitative Strategy \citep{db2016} (Figure 30 therein). $\halpha_i$ corresponds to (the $t$-values of) the alpha of each factor.}
\label{tbl_factors}
\footnotesize
\begin{tabular}{|c|l|l|l|l|}
\hline
No & \multicolumn{1}{c|}{Factor} & \multicolumn{1}{c|}{Description}  & \multicolumn{1}{c|}{Category} & \multicolumn{1}{c|}{$\halpha_i$}           \\ \hline
1  & Book-value to Price Ratio                           & Net Asset/Market Value & Value & 2.838                          \\ \hline
2  & Earnings to Price Ratio                           & Net Profit/Market Value & Value & 2.150                       \\ \hline
3  & Dividend Yield                           & Dividend/Market Value & Value  & 0.771                          \\ \hline
4  & Sales to Price Ratio                           & Sales/Market Value & Value & 0.981                 \\ \hline
5  & Cash Flow to Price Ratio                          & 
Operating cash flow/Market Value    & Value & 1.802              \\ \hline
6  & Return on Equity                           & Net Profit/Net Asset & Quality & 3.518 \\ \hline
7  & Return on Asset                           &  Net Operating Profit/Total Asset  & Quality & 3.843\\ \hline
8  & Return on Invested Capital                          & Net Operating Profit After Taxes/(Liabilities with interest + Net Asset) & Quality & 3.973 \\ \hline
9  & Accruals                       & -(Changes in Current Assets and Liability-Depreciation)/Total Asset & Quality & 0.913\\ \hline
10 & Total Asset Growth Rate                        & Change Rate of Total Assets from the previous period  & Growth & 3.995 \\ \hline
11 & Current Ratio                          & Current Asset/Current Liability  & Quality &  1.151 \\ \hline
12 & Equity Ratio                        & Net Asset/Total Asset & Quality & 0.094 \\ \hline
13 & Total Asset Turnover Rate                        & Sales/Total Asset  & Quality & 0.005 \\ \hline
14 & CAPEX Growth Rate                      & Change Rate of Capital Expenditure from the previous period & Growth & 2.417 \\ \hline
15 & EPS Revision (1 month)                & 1 month Earnings Per Share (EPS) Revision & Sentiment & 2.648  \\ \hline
16 & EPS Revision (3 month)                & 3 month Earnings Per Share (EPS) Revision & Sentiment & 2.031\\ \hline
17 & Momentum (1 month)                      & Stock Returns in the last month & Risk/Reversal  & 1.199\\ \hline
18 & Momentum (12-1 month)            & Stock Returns in the past 12 months except for last month & Momentum & 1.709 \\ \hline
19 & Volatility                       & Standard Deviation of Stock Returns in the past 60 months & Risk/Reversal & 2.263 \\ \hline
20 & Skewness                            & Skewness of Stock Returns in the past 60 months & Risk/Reversal  & 0.720 \\ \hline
\end{tabular}
\end{table*}%

\section{Limitations}
\label{sec_liminations}

\begin{itemize}
\item \textbf{Serial correlation:} We have assumed independence of  samples across time (Assumption 2.2). As discussed in \cite{Fama1965,famafrench2010,cont2001empirical}, the literature on stock returns often ignores this auto-correlation\footnote{\cite{cont2001empirical} states the absence of auto-correlation as one of the stylized statistical properties of asset returns.}. Moreover, the auto-correlation of the real data is very weak, shown in Appendix \ref{sec_serialcorrelation}, which justifies our i.i.d. assumption.
\item \textbf{Zero-mean factors:} We assume each feature is mean zero, which is consistent with the arbitrage pricing theory (c.f. Eq.~(2) in \cite{ross1976arbitrage}). Although we consider the method to be useful even in the case where this assumption does not hold, Assumption 2.3 is crucial in deriving Theorems \ref{thm_marginaldist} and \ref{thm_resboot}. 
Without Assumption 2.3, the over/under-fitting of the three-factor dimensions affects the estimation of $\hat{a}_i$.
In the proof of Theorem \ref{thm_marginaldist} (see notation therein), we derive  
\begin{align}
\hat{a}_i 
&= (\bX_1^\top \bX_1)^{-1} \bX_1^{\top} (\br_i - \bX_{2:}^\top (\hat{\bbeta}_i)_{2:}),
\end{align}
and the effect of the three-factor regression term $\bX_{2:}^\top (\hat{\bbeta}_i)_{2:}$ does not cancel out without Assumption 2.3.
\item \textbf{Asymptotics:} Availability of the null distribution (Assumption 3.2) is based on the asymptotic property of bootstrap samples as $T \rightarrow \infty$ (Theorem \ref{thm_resboot}). We consider this reasonable because $T$ is moderately large ($> 200$), and a large portion of bootstrap literature provides asymptotic theory. Moreover, many variants of asymptotic statistical testings are widely used (e.g., the Chi-squared test for contingency tables is frequently used instead of Fisher's exact test).
\end{itemize}

\section{Computational complexity of our proposed method}

Algorithm \ref{alg_proposed} conducts doubly nested bootstrapping. It samples  $\hbalpha^{(v)}$ for $V$ times, and for each $\hbalpha^{(v)}$ it estimates the FDR with $W$ bootstrap samples. Therefore, its computational complexity is $O(VW)$ to $V,W$. Lemma \ref{lem:numhypo} implies $V = O(1/q)$ samples are sufficient to avoid the underestimation of $N_0$ with probability at least $1 - q$. Theorem \ref{thm:main_prac} implies $W = \tilde{O}(1/q^2)$ samples\footnote{$\tilde{O}$ ignores a polylogarithmic factor.} suffice to bound the estimation error of the FDR up to $o(q)$. Therefore, $O(VW) = \tilde{O}(1/q^3)$ calculation of the FDR is required. Note that each calculation of the FDR is $\tilde{O}(N)$ to the number of hypotheses $N$ since it only requires a step-up procedure for the sorted $p$-values. 
Parallelizing Algorithm \ref{alg_proposed} with respect to $V$ or $W$ is easy.

\section{List of factors}

Table \ref{tbl_factors} lists the twenty factors used in the simulation.

\section{Full simulation results}

Table \ref{tbl_results_synth_all} shows the results of the rest of the synthetic simulations.
\begin{itemize}
    \item BH, BKY, BY, YB, and DDB always control FDR. Among them, DDB performs best.
    \item Storey, Storey-A, and DDBA sometimes has larger FDR than $q = 0.05$. Compared with Storey and Storey-A, DDBA is more robust to correlation. In the worst instance of complete null, Storey's methods has FDR of $> 0.15 = 3q$, whereas DDBA has FDR at most $0.08 < 0.10 = 2 q$.
\end{itemize}

\clearpage

\begin{table*}[]
\begin{footnotesize}
\begin{center}
\caption{Results of supplemental synthetic simulations. $\pm$ in FDR indicates the standard two-sigma confidence interval. Red characters indicate FDRs larger than $0.05$, which implies the failure in controlling FDR.
}
\label{tbl_results_synth_all}
\vspace{1em}
\begin{tabular}{llllllllllll}
                & \multicolumn{2}{l}{LSU}                                                 &  & \multicolumn{2}{l}{Storey} &  & \multicolumn{3}{l}{Boostrapping} &  & Fixed  \\ \cline{2-3} \cline{5-6} \cline{8-10} \cline{12-12} 
                & BH                   & BKY &   & Storey         & Storey-A         &  & YB           & \textbf{DDB} & \textbf{DDBA}          &  & Single \\ \hline

\\
\multicolumn{5}{l}{Scenario 4: $\sigma_{i,j}=0.0, \pi_0 = 1.0$} & & & & & & & \\ \hline
Thr-p&0.0000&0.0000&&0.0000&0.0723&&0.0000&0.0000&0.0001&&0.0311\\
\# of Rej&0.06&0.06&&0.06&3.96&&0.04&0.04&0.09&&2.49\\
FDR&0.0505&0.0500&&0.0555&\textcolor{red}{0.1465}&&0.0345&0.0390&\textcolor{red}{0.0800}&&\textcolor{red}{0.9095}\\ 
&$\pm$0.0098&$\pm$0.0097&&$\pm$0.0102&$\pm$0.0158&&$\pm$0.0082&$\pm$0.0087&$\pm$0.0121&&$\pm$0.0128\\ \hline

\\
\multicolumn{5}{l}{Scenario 5: $\sigma_{i,j}=0.5, \pi_0 = 1.0$} & & & & & & & \\ \hline
Thr-p&0.0001&0.0003&&0.0110&0.1085&&0.0000&0.0006&0.0019&&0.0237\\
\# of Rej&0.18&0.22&&1.61&6.69&&0.08&0.32&0.70&&2.56\\
FDR&0.0345&0.0345&&\textcolor{red}{0.0850}&\textcolor{red}{0.1695}&&0.0290&0.0345&\textcolor{red}{0.0645}&&\textcolor{red}{0.6365}\\
&$\pm$0.0082&$\pm$0.0082&&$\pm$0.0125&$\pm$0.0168&&$\pm$0.0075&$\pm$0.0082&$\pm$0.0110&&$\pm$0.0215\\ \hline
\\
\multicolumn{5}{l}{Scenario 6: $\sigma_{i,j}=0.9, \pi_0 = 1.0$} & & & & & & & \\ \hline
Thr-p&0.0006&0.0012&&0.0619&0.0747&&0.0003&0.0029&0.0061&&0.0083\\
\# of Rej&0.69&0.78&&8.54&9.30&&0.46&1.49&2.27&&2.15\\
FDR&0.0165&0.0160&&\textcolor{red}{0.1715}&\textcolor{red}{0.1860}&&0.0210&0.0335&0.0490&&\textcolor{red}{0.1871}\\
&$\pm$0.0057&$\pm$0.0056&&$\pm$0.0169&$\pm$0.0174&&$\pm$0.0064&$\pm$0.0080&$\pm$0.0097&&$\pm$0.0174\\ \hline

\\
\multicolumn{5}{l}{Scenario 7: $\sigma_{i,j}=0.0, \pi_0 = 0.25$} & & & & & & & \\ \hline
Thr-p&0.0009&0.0009&&0.0022&0.1518&&0.0002&0.0009&0.0045&&0.0409\\
\# of Rej&1.34&1.31&&2.04&10.39&&0.69&1.35&2.89&&8.32\\
FDR&0.0155&0.0155&&0.0219&\textcolor{red}{0.0638}&&0.0106&0.0147&0.0286&&\textcolor{red}{0.0704}\\
&$\pm$0.0042&$\pm$0.0042&&$\pm$0.0044&$\pm$0.0055&&$\pm$0.0041&$\pm$0.0041&$\pm$0.0045&&$\pm$0.0041\\ \hline

\\
\multicolumn{5}{l}{Scenario 8: $\sigma_{i,j}=0.5, \pi_0 = 0.25$} & & & & & & & \\ \hline
Thr-p&0.0021&0.0041&&0.0331&0.2066&&0.0006&0.0042&0.0125&&0.0386\\
\# of Rej&2.44&2.89&&5.72&13.72&&1.48&2.91&4.73&&8.14\\
FDR&0.0099&0.0122&&0.0251&\textcolor{red}{0.0682}&&0.0068&0.0105&0.0202&&\textcolor{red}{0.1062}\\
&$\pm$0.0033&$\pm$0.0034&&$\pm$0.0040&$\pm$0.0057&&$\pm$0.0030&$\pm$0.0027&$\pm$0.0040&&$\pm$0.0102\\ \hline

\\
\multicolumn{5}{l}{Scenario 9: $\sigma_{i,j}=0.9, \pi_0 = 0.25$} & & & & & & & \\\hline 
Thr-p&0.0030&0.0074&&0.0705&0.1964&&0.0025&0.0081&0.0184&&0.0292\\
\# of Rej&3.52&4.28&&10.17&16.44&&3.38&4.59&6.66&&8.22\\
FDR&0.0083&0.0127&&0.0438&\textcolor{red}{0.0786}&&0.0091&0.0133&0.0212&&\textcolor{red}{0.0651}\\
&$\pm$0.0022&$\pm$0.0025&&$\pm$0.0042&$\pm$0.0052&&$\pm$0.0030&$\pm$0.0027&$\pm$0.0033&&$\pm$0.0096\\ \hline

\end{tabular}
\end{center}
\end{footnotesize}
\end{table*}%

\clearpage

\section{Serial correlation}
\label{sec_serialcorrelation}

\begin{figure}
\vspace{2em}
\centering
\includegraphics[scale=0.75]{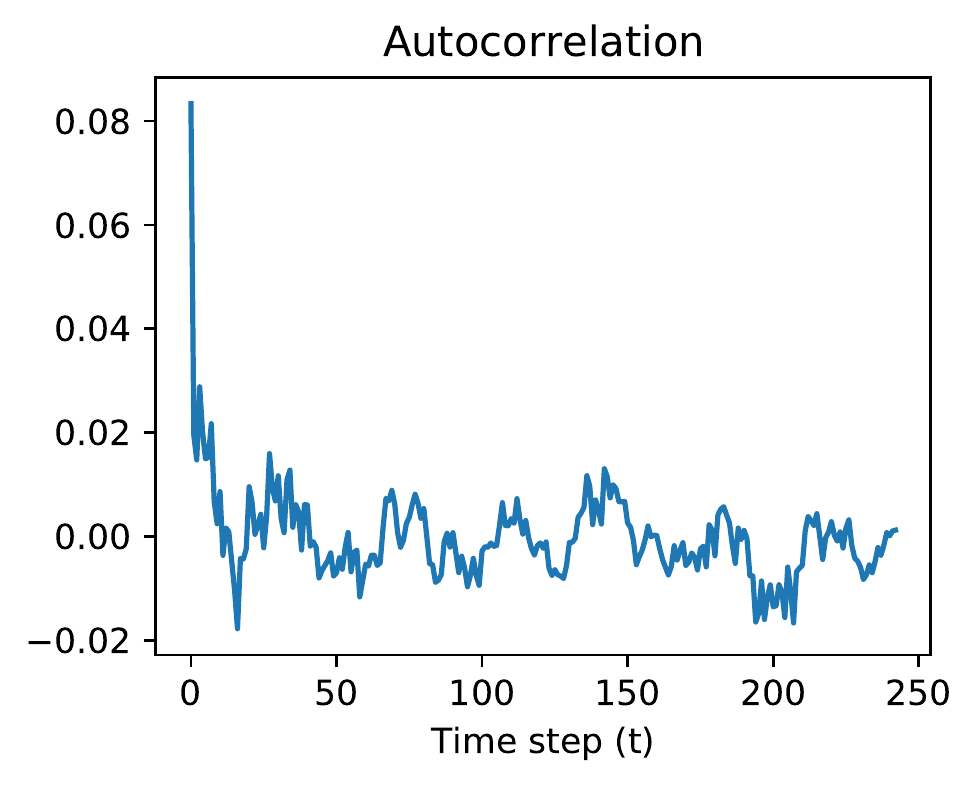}
\caption{Average autocorrelation in $\hat{a}_{i,t}$ averaged over twenty major factors listed in Table \ref{tbl_factors}.}
\label{fig:autocorrelation}
\end{figure}

Figure \ref{fig:autocorrelation} is the average auto-correlation of the returns from the twenty major factors in MSCI USA Index data (Section \ref{subsec:simreal}), in which there is hardly any auto-correlation.

\textbf{Future work under auto-correlation:}
The literature of bootstrapping states that the standard bootstrap method is inconsistent under serial correlation. Instead, when the time series is alpha-mixing (i.e. decaying auto-correlation), block bootstrap methods are consistent (Theorem 3.1 and 3.2 in \cite{lahiri2003resampling}).

\section{Real data: null distribution}

Figure \ref{fig:tvalue_null} shows the distribution of the null distribution. It is clear that the shape of the null (if marginalized for each $i$) is reasonably similarly to $\Tdist_{T-4}$.
\begin{figure}
\vspace{2em}
\centering
\includegraphics[scale=0.75]{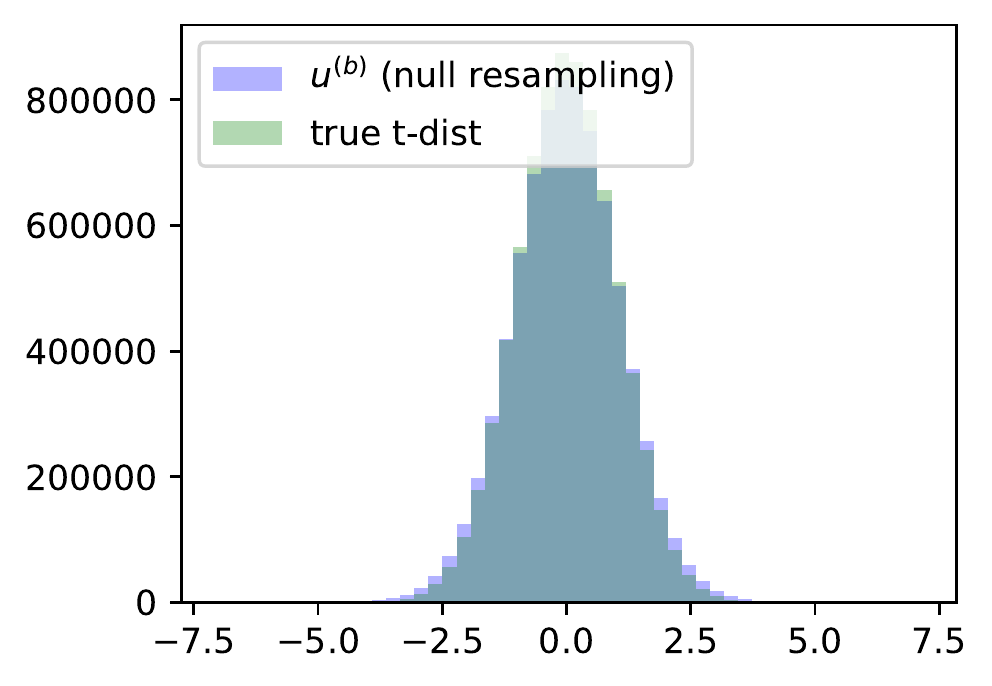}
\caption{Distribution of resampled residuals $(u_i^{(b)})_{i\in[N],b\in[B]}$ (``$u^{(b)}$ (null resampling)'') and random values drawn from the true $t$-distribution $\Tdist_{T-4}$ (``true $t$-dist''). Two distributions reasonably match. The resampled values are slightly fat-tailed than the true $t$ DGP.}
\label{fig:tvalue_null}
\end{figure}

\section{Proofs}\label{sec_proofs}


\begin{proof}[Proof of Theorem \ref{thm_marginaldist}]
The derivation of $\Normal(\bbeta_i, \sigma_i^2 (\bX^\top \bX)^{-1})$ is found in standard textbooks of econometrics (e.g., Section E-2 in \cite{Wooldridge}).
Let $(\bX_1, \bX_{2:}) = \bX$ be the first column (intercept term) and the rest (FF3 term). Namely, each row of $\bX_1$ is constant $1$ and each row of $\bX_{2:}$ is $(MKT_t, SMB_t, HML_t) \in \Real^3$. Let $((\hat{\bbeta}_i)_{1}, (\hat{\bbeta}_i)_{2:}) = \hat{\bbeta}_i \in \Real^{1+3}$ be defined in the same way.
The explicit formula of $\hat{a}_i$ is derived by using standard discussion on the partitioned regression model \footnote{Eq.~(23) in \url{https://www.le.ac.uk/users/dsgp1/COURSES/THIRDMET/MYLECTURES/2MULTIREG.pdf})}:
\begin{align}
\hat{a}_i &:= (\hat{\bbeta}_i)_1 \\
&= (\bX_1^\top \bX_1)^{-1} \bX_1^{\top} (\br_i - \bX_{2:}^\top (\hat{\bbeta}_i)_{2:})\\ 
&\hspace{-3em} \text{\ \ \ (by standard discussion on the partitioned regression model)}\\
&= (1/T) \sum_{t\in[T]} r_{i,t} \\ 
&\hspace{-3em} \text{\ \ \ (by the fact that each row of $\bX_{2:}$ has sum zero and $(\bX_1^\top \bX_1)^{-1} \bX_1^{\top} \bv = (1/T) \sum_{t\in[T]} v_t$ for a vector $\bv \in \Real^T$)} \\
&=a_i + \frac{1}{T} \sum_{t \in [T]} \hat{\eps}_{i, t} \\
&\hspace{-3em} \text{\ \ \ (by the fact that each row of $\bX_{2:}$ has sum zero)} \\
\end{align}
from which it is easy to see that
\begin{align}
\Ex[\hat{a}_i] &= a_i\\
T \Cov(\hat{a}_i, \hat{a}_j) &= \Cov(\eps_{i,1}, \eps_{j,1}).
\end{align}
\end{proof}

\begin{proof}[Proof of Theorem \ref{thm_resboot}]
Let $\ES(T) = [T] \times [T] \times ... \times [T]$ be the set of all the possible ways to obtain $T$ samples with replacements ($|\ES(T)| = T^T$).
For a random variable $r = r(\ET_b)$, let 
\begin{equation}\label{ineq_defboot}
\Ex_{\boot}[r(\ET_b)] := \frac{1}{|\ES(T)|} \sum_{\ET_b \in \ES(T)} r(\ET_b).
\end{equation}
By the definition of OLS with intercept term, we have 
\begin{equation}\label{ineq_defzeroavg}
\frac{1}{T} \sum_{t=1}^T \hat{\eps}_{i,t} = 0.
\end{equation}
Since $\hat{e}_i^{(b)}$ is a resampled residual,
\begin{equation}
\Ex_{\boot}[\hat{e}_i^{(b)}] = \frac{1}{|\ES(T)|} \sum_{\ET_b \in \ES(T)} \left( \frac{1}{T} \sum_{t\in \ET_b} \hat{\eps}_{i,t}\right)
= \frac{1}{T} \sum_{t=1}^T \hat{\eps}_{i,t} = 0,
\end{equation}
which is Eq.~\eqref{ineq_bootstrap_mean}.

We next derive Eq.~\eqref{ineq_bootstrap_cov}.
We have
\begin{align}
\hat{\eps}_{i,t} - \eps_{i,t} 
&= \bx_t^\top (\hat{\bbeta}_i - \bbeta_i) \\
&\le |\bx_t| |\hat{\bbeta}_i - \bbeta_i|.
\end{align}
Theorem \ref{thm_marginaldist} states that $\hat{\bbeta}_i - \bbeta$ is distributed as $\Normal(0, \sigma_i^2/(\bX^\top \bX)^{-1})$. 
Letting $||\bX||$ be constant, with probability $\log(1/\delta)$, 
\begin{equation}\label{ineq_epsest_bound}
|\hat{\eps}_{i,t} - \eps_{i,t}| = O(\log(1/\delta)/T) 
\end{equation}
for all $t \in [T]$.

Since $\hat{e}_i^{(b)}$ is a bootstrap average of $T$ samples from $(\hat{\eps}_{i,t})_{t \in [T]}$ and $\hat{\eps}_{i,t}$ is zero-mean,
\begin{align}
 T \Cov_{\boot}(\hat{e}_i^{(b)}, \hat{e}_j^{(b)}) 
 &= 
 \frac{1}{T} \Ex_{\boot}\left[\left(\sum_{t \in \ET_b} \hat{\eps}_{i,t}\right) \left(\sum_{t \in \ET_b} \hat{\eps}_{j,t}\right)\right] \text{\ \ \ (by definition)}\\
 &= 
 \frac{1}{T} \Ex_{\boot}\left[\sum_{t \in \ET_b} \hat{\eps}_{i,t} \hat{\eps}_{j,t}\right]\\
& \text{\ \ \ (by i.i.d. property of sampling with replacement and Eq.~\eqref{ineq_defzeroavg})}\\
 &= 
\frac{1}{T} \sum_{t \in [T]} \hat{\eps}_{i,t} \hat{\eps}_{j,t} \\
 &= \Cov(\eps_{i,1}, \eps_{j,1}) + O\left(\frac{\log(1/\delta)}{T}\right),
 \text{\ \ \ (by Eq.~\eqref{ineq_epsest_bound})}
\end{align}
which concludes Theorem \ref{thm_resboot}.
\end{proof}

\begin{proof}[Proof of Lemma \ref{lem:numhypo}]
Let 
\begin{align}
\mone(\bu', \bu'') &= \sum_{i \in \Hzero} \Ind\left[u_i' \ge u_i''\right]\\
\mtwo(\bu', \bu'') &= \sum_{i \in \Hzero} \Ind\left[|u_i'| \ge |u_i''|\right]
\end{align}
for $\bu', \bu'' \in \Real^N$.
It is easy to check that, $m(\bu', \bu'') \ge N_0/2$ or $m(\bu'', \bu') \ge N_0/2$ always holds for $m \in \{\mone, \mtwo\}$. 
Let an event be $\ED(\bu', \bu'') = \{m(\bu', \bu'') \ge N_0/2\}$.
We may interpret $\ED(\bu', \bu'')$ as a result of a ``duel'' between two ``players'' $\bu'$ and $\bu''$, where at least one of $\bu'$ or $\bu''$ will win. The goal of the proof is to bound the probability that a ``random player'' $\bu'$ is near insurmountable (very hard-to-beat), is not very high, which guarantees that there is at least one instance, such that $\ED(\bu^{(v)}, \balpha) \subseteq D(\bu^{(v)}, \bu)$, exists for $\bu \sim \EU$ with high probability.
In the following, we define function $s(\bu')$ that represents ``skill'': The ratio of $\bu''$ that beats $\bu'$. Smaller $s$ represents a player that is harder to beat. We can say that the player with skill $q$ or less is at most $2q$. 
Let 
\begin{equation}
s(\bu') = \Pr_{\bu'' \sim \EU}[\ED(\bu'', \bu')].
\end{equation}
and 
\begin{align}
\ES(q) &= \{\bu' \in \EU: s(\bu') \le q\} \\
S(q)   &= \Pr_{\bu' \sim \EU}[\bu' \in \ES(q)]. 
\end{align}
For any $q \in (0,1)$, we prove that $S(q) \le 2q$ by contradiction. Assume that $S(q)= 2q + \eps$ for some $q, \eps > 0$. 
Let $f(\bu')$ be the PDF of $\bu' \sim \EU$.
By definition, $\bu' \in \ES(q)$ implies
\begin{align}\label{ineq:ESdef}
q 
&\ge \int \Ind[\ED(\bu'', \bu')] df(\bu'')\\
&\ge \int \Ind[\ED(\bu'', \bu'), \bu'' \in \ES(q)] df(\bu'')
\end{align}
and thus we have
\begin{align}
\lefteqn{
(2q+\eps)q
}\\
&= q \int \Ind[\bu' \in \ES(q)] df(\bu') \text{\ \ \ \ (by $S(q)= 2q + \eps$)} \\
&\ge \int \int \Ind[\ED(\bu'', \bu'), \bu', \bu'' \in \ES(q)] df(\bu'') df(\bu') \\
&= (2q+\eps)^2/2 \\
& \text{\ \ \ \ (at least one of $\ED(\bu'', \bu')$ or $\ED(\bu', \bu'')$ always hold)},
\end{align}
which contradicts.

Let $\ED^c$ be a complementary event of $\ED$.
We have
\begin{align}
\lefteqn{
1 - \Pr\left[ \bigcup_{v \in [V]}\left\{ m(\bu^{(v)}, \hbalpha) \ge N_0/2\right\} \right]
}\\
&\le 1 - \Pr\left[ \bigcup_{v \in [V]}\left\{ m(\bu^{(v)}, \bu) \ge N_0/2\right\} \right]
\\
&= \Pr\left[ \bigcap_{v \in [V]}\left\{ m(\bu^{(v)}, \bu) < N_0/2\right\} \right] \\
&= \Pr\left[ \bigcap_{v \in [V]}\left\{ \ED^c(\bu^{(v)}, \bu)\right\} \right] \\
&\le \int_0^1 (1 - q)^{V} dS(q) \\
&\text{\ \ (by independence of bootstrap samples)} \\
&\le \int_0^{1/2} (1 - q)^{V} 2 dq \text{\ \ \ \ (by $S(q) \le 2q$)}\\
&= \frac{2}{V+1} [(1 - q)^{V+1}]_0^{1/2} \le \frac{2}{R_1+1},
\end{align}
which implies Eq.~\eqref{ineq:numhypo}.
\end{proof}

\begin{proof}[Proof of Theorem \ref{thm:main_prac}]
Let $e = \sqrt{\frac{\log((V S)/\delta)}{2 W}}$.
A union bound of Lemma \ref{lem_empfdr} with $\delta = \delta/(VS)$ over $v = 1,2,\dots,V$ and all the $S$ candidates $c_q^{(v,s)}$ of $c_q^{(v)}$ yields
\begin{equation}\label{ineq:estfdrdiff}
\forall\,v \in [V]\,s\in[S]\, |\FDR(\hbalpha^{(v)}, c_q^{(v,s)}) - \hatFDR(\hbalpha^{(v)}, c_q^{(v,s)})| \le e
\end{equation}
with probability at least $1-2\delta$.
Let 
\begin{align}
C_q(e)  &= \min_{v \in [V]} \sup\left\{c_q\in(0,1): \FDR(\hbalpha^{(v)}, c_q) \le q/2 + e\right\} \\
c_q^*(e) &= \sup\left\{c_q\in(0,1): \FDR(\balpha, c_q) \le q + 2 e\right\}
\end{align}
Then, we have
\begin{align}
\min_{v\in[V]} c_q^{(v)} 
&\le C_q(e) \text{\ \ \ \ (by \eqref{ineq:estfdrdiff})}\\
&\le c_{q'}^*(e) \text{\ \ \ \ (by \eqref{ineq:assm_existhard} with $q' = q + 2e$)},
\end{align}
which, together with the non-decreasing property of FDR, yields
\begin{align}
\FDR\left(\balpha, \min_{v\in[V]} c_q^{(v)}\right) 
&\le \FDR(\balpha, C_q(e)) \\
&\le \FDR(\balpha, c_{q'}^*(e)) = q + 2 e
\end{align}
which completes the proof.
\end{proof}

\section{Instruction for reproducing simulation results}


We conducted each of our simulations on a standard Xeon server with $> 20$ cores. We ran Python3.9.1 with numpy, scipy, pandas, matplotlib, multiprocessing, and joblib libraries. 

\subsection{Reproducing synthetic data simulations}

Each of our nine synthetic scenarios are reproduced by 
\begin{lstlisting} 
python fdr_factors.py -r 2000 -s X -c 1  > scenarioX_r2000.txt 
python printresults.py scenarioX_r2000.txt
\end{lstlisting}
where $X \in \{1,2,3,\dots,9\}$.
It takes about 20 hours for our server to complete each of the scenarios of 2,000 runs. The computation time per run is below 1 minute, which is feasible in many real use cases.

\subsection{Reproducing real data simulations}

Our real-world data simulation is reproduced by
\begin{lstlisting} 
python fdr_factors.py -r 100 -s 0 -a data/US_tvalues.csv -x  \\
          data/US_tvalues_null.csv -c 1 > real780_r100.txt
python printresults.py real780_r100.txt
\end{lstlisting}
It takes $<4$ hours for our server to complete the above code of 100 runs (less than $3$ minutes per run).

\subsection{Reproducing real data autocorrelation}

Figure \ref{fig:autocorrelation} can be reproduced by 
\begin{lstlisting} 
python autocorrelation.py
\end{lstlisting}

\end{document}